\documentclass[acmsmall, 10pt]{acmart}

\usepackage{soul}
\newtheorem{theorem}{Theorem}[section]
\usepackage{subfig}

\setcopyright{rightsretained}

\newtheorem{lemma}[theorem]{Lemma}

\newtheorem*{proof*}{Proof}
\newtheorem{definition}[theorem]{Definition}
\newtheorem*{theorem*}{Theorem}

\newtheorem{remark}[theorem]{Remark}
\newtheorem{notation}[theorem]{Notation}

\begin{document}

\title{On The Convexity Of The Effective Reproduction Number}
	\author{Jhonatan Tavori}
	\affiliation{%
		\institution{
			Blavatnik School of Computer Science, Tel-Aviv Univeristy}
		\country{Israel}
	}
	\email{jhonatan.tavori@cs.tau.ac.il}
	
	\author{Hanoch Levy}
	\affiliation{%
		\institution{
			Blavatnik School of Computer Science, Tel-Aviv Univeristy}
		\country{Israel}
	}
	\email{hanoch@tauex.tau.ac.il}
\begin{abstract}
	
In this study we analyze the evolution of the effective reproduction number, $R$, through a \textit{SIR} {spreading} process in {\it heterogeneous} networks; Characterizing its decay process allows to analytically study the effects of countermeasures on the progress of the {virus} under heterogeneity, and to optimize their policies.

A striking result of recent studies has shown that heterogeneity across {nodes/individuals} (or, super-spreading) may have a drastic effect on the spreading process progression, which may cause a non-linear decrease of $R$ in the number of infected individuals.
We account for heterogeneity and analyze the stochastic progression of the spreading process. We show that the decrease of $R$ is, in fact, \textit{\textbf{convex}} in the number of infected individuals, where this convexity stems from heterogeneity. The analysis is based on establishing \textit{stochastic monotonic relations} between the susceptible populations in varying times of the spread.

We demonstrate that the convex behavior of the effective reproduction number affects the performance of countermeasures used to fight a spread of a virus. {The results are applicable to the control of virus and malware spreading in computer networks as well.} We examine numerically the sensitivity of the Herd Immunity Threshold (HIT) to the heterogeneity level and to the chosen policy.

\end{abstract}

\maketitle

\section{Introduction}

In 1923 Topley and Wilson described experimental epidemics in which the rising prevalence of immune individuals would end an epidemic. They named this phenomenon as ”Herd-Immunity” \cite{topley1923spread}. 
Herd immunity is achieved when the effective reproduction number $R$ (the expected number of secondary infections produced by an infected node) reduces below $1$ and the number of infection cases diminishes. 

The fraction of the population which is required to contract the virus (or get immune in another way, e.g., get vaccinated) in order to reduce $R$ below 1 is known as the Herd Immunity Threshold (HIT).
Estimations of the COVID-19 HIT are used by governments worldwide in determining policies to fight the ongoing pandemic.
Thus, analyzing and predicting the behavior of the reproduction number, $R$, is important for {infectious disease control and immunization.}

The objective of this research is to {study and characterize the behavior of the effective reproduction number} in an infectiousness {\it heterogeneous} population, in order to {optimize disease-blocking strategies}.

A key factor inherent to our model and analysis is {\it heterogeneity of infectiousness and susceptibility}. 
Numerous real-world networks, including the human interaction network, have been known to be characterized by heterogeneity of connectivity and interaction between their individuals. 
The ongoing COVID-19 pandemic has made it possible to observe such a heterogeneity phenomenon conspicuously, as distinct individuals are infectious (likely to infect others) and susceptible (likely to become infected themselves) in various degrees. 
The estimates for the COVID-19 pandemic fits this property and asserts that between 5\% to 10\% of the infected individuals (i.e., the \textit{"super-spreaders"}) cause more than 80\% of the secondary infections \cite{endo2020estimating,miller2020full}.

In striking results of recent studies \cite{oz2021heterogeneity,britton2020mathematical} heterogeneity across individuals has been shown to possibly have a drastic effect on the viral process progression, and to \textit{reduce} dramatically the number of infected individuals prior to reaching {herd immunity}.
Further, it was shown that the decrease of $R$ in heterogeneous systems is not necessarily linear in the number of infected individuals (unlike the linearity under homogeneous populations). {Roughly speaking, the existence of nodes with high level of infectiousness/susceptibility (i.e., \textit{super-spreaders}) yields that they stochastically tend to get infected, develop immunity, and "leave the game" in an early stage of the epidemic process.}

A key question motivating this research, is {how disease-blocking  strategies and their performance are affected by heterogeneity?
To answer these questions we} analyze the stochastic progression of the spreading process during a natural-evolution (i.e., with no countermeasures) of the process and establish that the decrease of $R$ in heterogeneous networks is, in fact, \textit{convex} in the number of infected individuals.

The convexity of $R$ stems directly from the phenomenon of super-spreading (heterogeneity), where the proof is based on establishing stochastic monotonic relations between the susceptible populations in varying times of the spread. 
The decay of $R()$, for various heterogeneity levels, is demonstrated in Figure \ref{Fig:DiffCVR}.

\begin{figure}[h]
	\centering
	\includegraphics[width=0.678\linewidth]{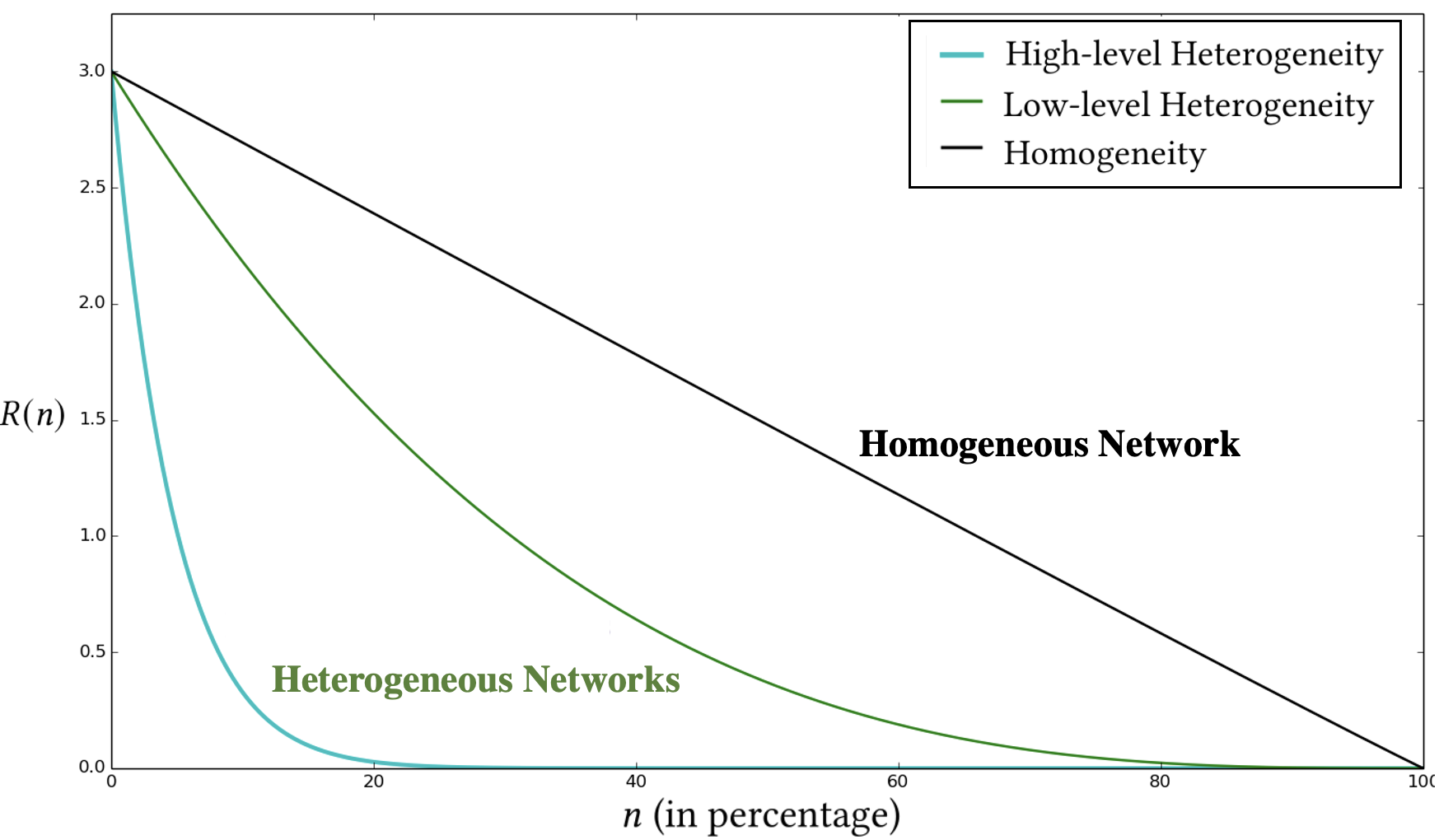}
	\caption{Decay of the effective reproduction number $R(n)$, as a function of $n$, for various heterogeneity levels.}
	\label{Fig:DiffCVR}
\end{figure}

{Having established the convexity property, we then use it to address operational problems and discuss the effect of the convexity property on infectious disease countermeasure policies. We also examine the sensitivity of the HIT to the heterogeneity level and to the chosen policy.}

We begin with the problem of vaccines allocation, where a limited number of vaccines need to be allocated to different regions (i.e., countries), such that the total number of infected individuals prior to reaching herd immunity will be minimized. 
We demonstrate that heterogeneity, and the convexity of $R$, affect the outcome of such vaccination processes as allocations which disregard heterogeneity and conduct optimization by allocating the vaccines only according to other parameters (i.e., the population sizes and basic reproduction number) might perform significantly worse comparing to the heterogeneity accounting optimizations.
{In addition, we use numerical simulations and show that minimization of the HIT is achieved by \textit{advancing} the vaccine administration timing, namely by giving them as early as possible}.

{We address lockdowns as well.}
It was recently revealed \cite{tavori2021super} that due to the effect of super-spreading in pandemics, lockdowns are very sensitive to heterogeneity and might be (perhaps contrary to naive intuition) even counter productive and \textit{increase} the HIT, if heterogeneity is not properly taken into account.
{Our results imply that any restriction on social interaction which will transform the network into a "semi-homogeneous" network will increase the HIT comparing to the natural spread of the disease.}

\section{The Model}
{We adopt a model} \cite{oz2021heterogeneity} that intrinsically accounts for heterogeneity of infectiousness and susceptibility. Such heterogeneity has been recognized and studied for quite a while; see e.g.,
\cite{rock2014dynamics, smith2005, lloyd2001viruses}. 
The spreading process follows the  \textit{Susceptible-Infective-Recovered} (SIR) model, a standard model of epidemic processes \cite{lloyd2001viruses,pastor2015epidemic,newman2003structure}, with the extension of \textit{Vaccinations}.

Let $A_0 = \{a_1, a_2, \dots, a_{N_0}\}$ be the set of the network's nodes (i.e., \textit{population}), where $\vert A_0 \vert = N_0$.
Our analysis begins with a certain number of infective nodes, where all others are assumed to be initially susceptible (\textbf{S}). 
As a result of an \textit{infection}, susceptible nodes become infective (\textbf{I}). 
After an infectious period, infected nodes "leave the game" (\textbf{R}), i.e., {stop spreading the virus} or being removed.
Additionally, susceptible nodes can move directly to being vaccinated (\textbf{V}), means that they can not get infected or infect others.

The spread across the network is indexed as a function of the number of nodes which contracted it. Namely, the event whereby the $n$th infection occurs is called the $n$th step of the disease. 
Thus, the number of steps is upper-bounded by $N_0$. 

\subsection{Stochastic Spreading Functions}

Each node $a$ is assigned with susceptibility and infectiousness parameters, $S(a)$ and $I(a)$; The values of $S$ and $I$ are in the range $[0, 1]$. These values accompany the individual $a$ throughout the entire process. 
$I(a_i)$ represents the probability of $a_i$ to spread the disease to others when it is in Infective mode. $S(a_j)$ represents the probability of $a_j$ to attract the infection. 
Assuming that $a_i$ is infective, the probability that it will infect the susceptible node $a_j$ is:
$$
\Pr[\text{$a_i$ infects $a_j$}] = I(a_i) \cdot S(a_j).
$$
Note that $S(a)$ and $I(a)$ implicitly include the social interaction level of $a$ (probability of meeting other individuals) as well as any personal or physical properties (such as biological features, {social distancing obedience} or others).

Throughout our analysis we will assume that no individual infects \textit{directly} a {significant} fraction of the population (while allowing heterogeneity of the spreading parameters)
\cite{britton2020mathematical, smith2005, oz2021heterogeneity}. 

The ensemble of the values over all $a \in A_0$ forms the distributions $\mathcal{S}_p$ and $\mathcal{I}_p$
\footnote{Practically -- these values can be produced by sampling a population. {Examples of how they were estimated can be found in} \cite{endo2020estimating} and \cite{smith2005}. If one wants to attribute them to a specific theoretical distribution, say Gamma (,) one may draw each of the values from that distribution and obtain $\mathcal{S}_p$ and $\mathcal{I}_p$ that approximate the theoretical Gamma very well (especially in practical situations where population sizes are in millions). }. $\mathcal{S}_p$ and $\mathcal{I}_p$ denote the \textit{susceptibility and infectiousness distributions} of the population (network). $ \text{Supp} (\mathcal{S}_p) = \text{Supp}(\mathcal{I}_p) = [0,1]$.

\noindent
\subsubsection*{Correlation of  susceptibility and infectiousness}
In many cases it is logical to assume that the susceptibility and {infectiousness} of each individual $a$ are equal, as susceptibility and infectiousness levels are both proportional to the level of interaction or other properties of $a$.
Yet, the results developed in this work are based on a significantly lighter assumption --  \textit{monotonicity of spreading}:  more susceptible individuals are stochastically more infectiousness. 
Formally, denote by:
\begin{notation}[Expected Conditional Infectiousness]
\begin{equation}
	\varphi(s) := \mathop{\mathbb{E}} _{a \in A_0}  \left[ I(a) \: \vert \: S(a) = s \right].
\end{equation}
\end{notation}
We say that $\mathcal{S}$ and $\mathcal{I}$ possesses {\textit{correlated monotonicity}
if their sampling is positively correlated such that $\varphi(s)$ is monotonically non-decreasing in $s$. 
Throughout this work, we assume that $\mathcal{S}$ and $\mathcal{I}$ possesses correlated monotonicity. Of course, the case of equality $I(a)=S(a)$ for any $a \in A_0$ is a special case of correlated monotonicity. The case where the infectiousness is independent of susceptibility (i.e., $\varphi(s_1) = \varphi(s_2)$ for any $s_1 \neq s_2$) is a special case as well.
	
\subsubsection*{Tracking the Spreading Distribution}
The {distribution} of the susceptibility and infectiousness parameters across the susceptible population may change throughout the spread, as infected/vaccinated nodes are removed from the (susceptible) population.

Recall that we index the spread as a function of the number of nodes which contracted the virus (i.e., the $n$th infection case is called the $n$th \textit{step}).
Let $\textbf{A}_n$ denote the susceptible population at step $n$ (a subset of the population at step $n-1$). $\textbf{A}_n$ is a random variable distributed over all possible scenarios of infection. 

\begin{notation}[Susceptibility Density at Step $n$]
		We denote by $\rho(\cdot, n)$  the {Probability Density Function (pdf)} of the susceptibility {of $\textbf{A}_n$}, the susceptible population at step $n$. Namely, if an individual (node) $a$ is picked \textit{at random} from the susceptible population at step $n$, its susceptibility $S(a)$ is distributed according to $\rho (\cdot,n)$.
\end{notation}
Note that $\rho(\cdot,0)$ is {the pdf of} the susceptibility of $A_0$, namely the pdf of $\mathcal{S}$. This implies that the distribution is a discrete distribution (with support equaling the set formed by {the} $N_0$ values drawn initially).  The support of $\rho(\cdot,n)$ is a subset of that of $\rho(\cdot,0)$, following the epidemic process (at which node are removed from the susceptible population).

\subsection{The Reproduction Number}

The \textit{basic reproduction number}, $R_0$, is a measure of how transferable a disease is. It is defined as the expected number of secondary cases produced by a single (typical) infection in a completely susceptible population (of size $N_0$).

As the spread continues, varying proportions of the population are recovered or removed (or vaccinated) at any given time. 
Hence, we will measure the \textit{effective reproduction number}, $R(n)$, which is defined as the expected number of infections directly generated by the $n$th infected individual
(e.g., $R_0 = R(0)$). 
Formally,
$$
R(n) = \mathbb{E} \left[ \text{Number of infections generated by the $n$th infected individual} \right] .
$$
where the expectation is taken over all individuals and all possible scenarios of 
infections.

As long as $R$ obeys $R > 1$, then the number of infections increases; When $R < 1$, the number of infection cases decreases, and the spreading process comes to an end (i.e., herd immunity is achieved).
The fraction of the population that contracted the spread before $R$ reaches $1$ is called the \textit{Herd Immunity Threshold} (or HIT).

\section{Convexity of the Effective reproduction Number}\label{secrconvex}

In this section we track the evolution of the effective reproduction number throughout a natural evolution (i.e., no vaccinations) of the spreading process.
Intuitively, under the SIR epidemic process $R(n)$ is continuously decreasing as infected nodes move from state \textbf{S} to state \textbf{I}, and then to state \textbf{R}, and the size of the susceptible population decreases.

In Subsection \ref{subsec:homog} we discuss homogeneous populations, and show that the decay of $R()$ in such networks is linear. In Subsection \ref{subsec:hetero} we treat the general case of arbitrary heterogeneous populations, and prove \textit{convexity}.
{Note} that prior studies observed that the decrease of $R(n)$ in heterogeneous populations is not linear,
since the pdf of the susceptibility, $\rho(s,n)$, may change throughout the spreading process. Yet, {convexity} was not established before. 
Formally, we will prove the following theorem:
\begin{theorem}\label{thm:rconvex}
	$R(n)$ is convex in $n$. That is, for any $n$:
	$$  R(n+1) - R(n+2) \le R(n) - R(n + 1).$$
\end{theorem}

\subsection{Homogeneous Populations}\label{subsec:homog}
In homogeneous populations, the susceptibility and infectiousness parameters are identical across the individuals of the  population. I.e., there exists values $\sigma$ and $\iota$ such that for any node $v \in A_0$, $S(v)=\sigma$ and $I(v)=\iota$. Hence, since all the nodes are identical, the probability of each node to be the $n$th infected is $\frac{1}{\vert A_0 \vert} = \frac{1}{N_0}$. 
Therefore,
\begin{equation}\label{eq:rhomod}
	R(n) = \sum_{v \in A_0}   \frac{1}{N_0} \cdot I(v) \cdot \mathbb{E} \left[ \sum_{u \in \textbf{A}_n} S(u)\right] =
	R_0 \cdot \frac{N_0 - n}{N_0}.
\end{equation}

Hence for any $n$, $R(n+1) - R(n)  = R(n+2) - R(n+1)$, and Theorem \ref{thm:rconvex} holds.

\subsection{Heterogeneous Populations}\label{subsec:hetero}

It was proved in \cite{oz2021heterogeneity} that in heterogeneous networks, it holds that:
\begin{lemma}\label{lem:rdef}
	The {value} of the effective reproduction number at step $n$ can be approximated by
	\begin{equation}\label{eq:rdef}
		R(n) \approx
		 (N_0 - n) \cdot \int \rho(s,n) \cdot s  \cdot \varphi(s) \; ds .\footnote{{Note that since $\rho(s,n)$ is in fact discrete (as its support is a finite set of value), one should technically treat the notation ($\int$) as the coresponding sum}.}
	\end{equation}
\end{lemma}

\begin{remark}[{Accuracy} of the approximation in Eq. (\ref{eq:rdef})]\label{rmrk:eq11}
	The approximation was based on assuming that one may infect itself (and adding this event to the count of expected infections, in the derivation of Eq. (\ref{eq:rdef})). Hence the error is given by the probability of such an event, yielding a bound on the relative error  by $O(\frac{\max_a S(a)}{\sum_b S(b)})$.
	Recall that we assume that no individual infects a significant fraction of the population, 
{as was assumed in} \cite{oz2021heterogeneity}, {implying that $O(\frac{\max_a S(a)}{\sum_b S(b)})$ is negligible}.
\end{remark}

We start by stating a relationship between $\rho(s,n)$ and $\rho(s,n+1)$. 

\begin{definition} (see \cite{muller2002comparison, wolfstetter1993stochastic})
	The pair of {probability density functions} $(\rho_1(s ,n), \rho_2(s , n))$ possesses the monotone likelihood ratio property (MLRP) if 
	$	\frac{\rho_1(s, n)}{\rho_2(s, n)}$
	is non-decreasing in $s$.
\end{definition}

\begin{lemma}\label{lem:ratiodec}
For any $n$, the pair $(\rho(s,n), \rho(s, n+1))$ possesses the monotone likelihood ratio property
\end{lemma}

\begin{proof}
Let $A_n$ be a realization of the susceptible population at step $n$. 
We start by deriving
$\Pr[v \in \textbf{A}_{n+1} \; \vert \; v \in A_n]$,
where $\textbf{A}_{n+1}$ is a random variable distributed over al the possible susceptible populations at step $n+1$, given $A_n$. I.e., according to the infection probability.

Since the likelihood of getting infected is proportional to the susceptibility value of the node, it holds that the probability that $v \in A_n$ is the next infected is
\begin{equation}\label{eq:eq10}
	\frac{S(v)}{\sum_{u \in A_n} S(u)}.
\end{equation}
Therefore,
\begin{equation}\label{eq:An1ratioAn}
	P_n(s) :=
	\Pr[v \in \textbf{A}_{n+1} \; \vert \; v \in A_n \text{ and } S(v) = s] 
	= \left(1 - \frac{s}{\sum_{u \in A_n} S(u)} \right).
\end{equation}
Recall that  $\rho(s,n)$ is the density function of the susceptibility of the (susceptible) population at step $n$. In other words, if a node $u$ is picked \textit{at random} at step $n$, its susceptibility is distributed according to $\rho (s,n)$. 
Eq. (\ref{eq:An1ratioAn}) can be used to develop the ratio between $\rho(s,n)$  and $\rho(s,n+1)$:
\begin{equation}\label{eq:rhorationidea}
	\rho(s, n + 1)
	= \frac{ \rho(s,n)  P_n(s) }{ \int \rho(\sigma , n) P_n(\sigma)  d\sigma}.
\end{equation}
According to Eq. (\ref{eq:An1ratioAn}), $P_n(s)$ is monotonically decreasing in $s$.
Hence, by Eq. (\ref{eq:rhorationidea}) the ratio $\frac{\rho(s,n+1)}{\rho(s,n)}$ is decreasing in $s$. Therefore, the ratio $\frac{\rho(s,n)}{\rho(s,n+1)}$ is non-decreasing, and we conclude the proof.
\end{proof}

\begin{figure}[h]
	\centering
	\includegraphics[width=0.7\linewidth]{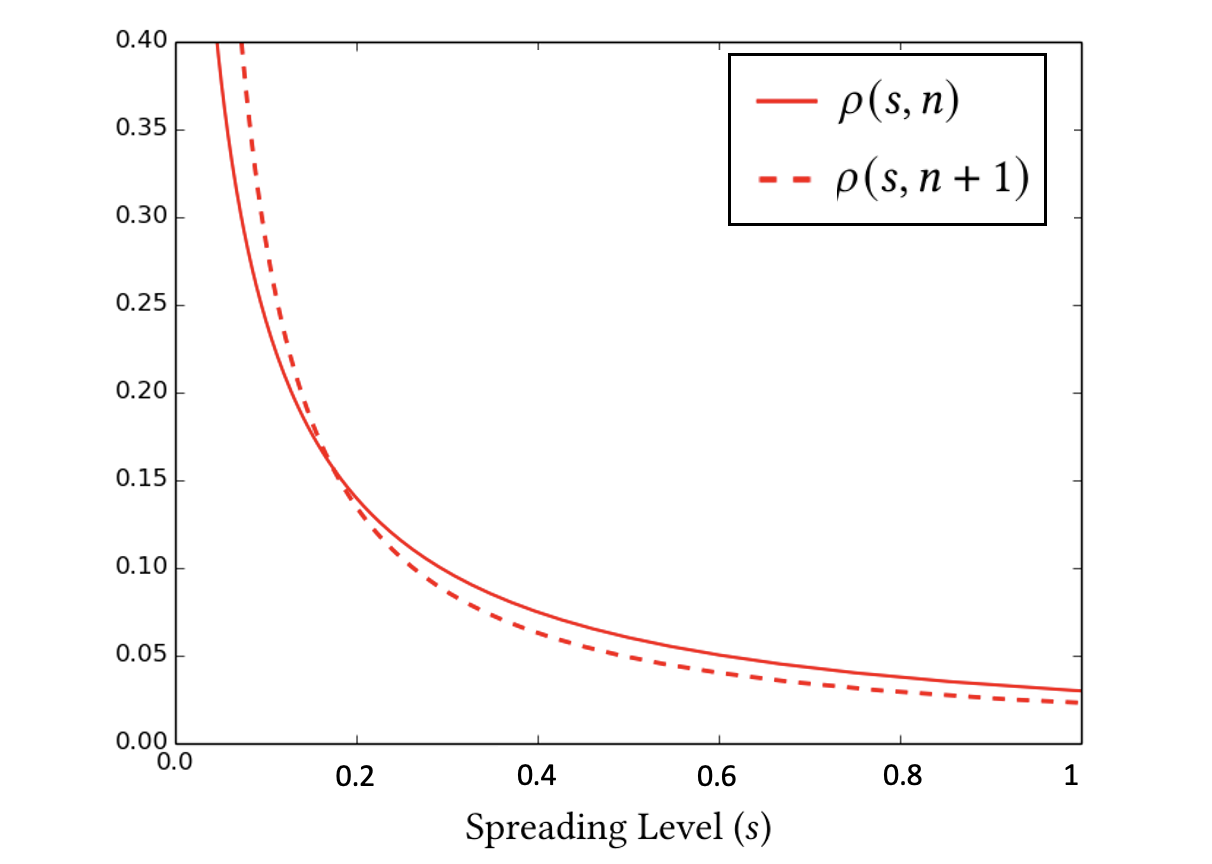}
	\caption{Demonstration of the relation between $\rho(s,n)$ and $\rho(s,n+1)$ as a function of $s$.}
	\label{Fig:mlrp}
\end{figure}

The relationship established in Lemma \ref{lem:ratiodec} is demonstrated in Figure \ref{Fig:mlrp}.
Using the lemma {and the following known theorems we prove Lemma} \ref{lem:rratiodec} and Theorem \ref{thm:rconvex}. 

\begin{definition} (see \cite{wolfstetter1993stochastic})
		{Let} $\textbf{X}_1, \textbf{X}_2$, be random variables. $\textbf{X}_1$ (first order) stochastically dominates $\textbf{X}_2$ (denoted by $\textbf{X}_1 \gtrapprox_{FSD} \textbf{X}_2$) if $\; \Pr[X_1 > z] \ge \Pr[X_2 > z]$ for all $z$. 
\end{definition}

\begin{theorem}(From \cite{wolfstetter1993stochastic})\label{thn:mlrpsausd}
	{Let} $\textbf{X}_1, \textbf{X}_2$, be random variables. Let $\rho_1(),  \rho_2()$ be their pdfs.
	If the pair $(\rho_1(), \rho_2())$ possesses the monotone likelihood ratio property then $\textbf{X}_1$ (first order) stochastically dominates $\textbf{X}_2$ as follows: $\textbf{X}_1 \gtrapprox_{FSD} \textbf{X}_2$. \cite{hadar1971stochastic,wolfstetter1993stochastic}. 
\end{theorem}

\begin{theorem}(From \cite{wolfstetter1993stochastic})\label{thn:sdsayexpec}
	{Let} $r()$ be an injective monotone function. If $\textbf{X}_1 \gtrapprox_{FSD} \textbf{X}_2$ then $\mathbb{E}\left[ r(\textbf{X}_1) \right] \ge \mathbb{E}\left[ r(\textbf{X}_2) \right]$.
\end{theorem} 

\begin{lemma}\label{lem:rratiodec}
{For any $0 < n_1 < n_2 < N_0$,}
$		\frac{R(n_2)}{R(n_1)} \le \frac{N_0 - n_2}{N_0 - n_1}$.
\end{lemma}

\begin{proof}[Proof of Lemma \ref{lem:rratiodec}]
	{By applying Lemma} \ref{lem:ratiodec} recursively over $n_1, n_1 + 1, \dots, n_2$, we have that $(\rho(s,n_1), \rho(s,n_2))$ possesses the monotone likelihood ratio property.
	By Theorems \ref{thn:mlrpsausd} and \ref{thn:sdsayexpec}, since $s \cdot \varphi(s)$ is monotone increasing in $s$, we have that 
	$\int \rho(s, n_1) s \varphi(s) ds \ge \int \rho(s, n_2) s \varphi(s) ds$.
	By Eq. (\ref{eq:rdef}),
	\begin{equation}
		\frac{R(n_2)}{R(n_1)} = 
		\frac{(N_0 - n_2) \cdot \int \rho(s, n_2) s \varphi(s) ds}
		{(N_0 - n_1) \cdot \int \rho(s, n_1) s \varphi(s) ds}
		\le
		\frac{N_0 - n_2}
		{N_0 - n_1}
	\end{equation}
{and we conclude the proof.}
\end{proof}

In order to prove Theorem \ref{thm:rconvex} we will use the following known theorem:

\begin{theorem}[From \cite{wijsman1985useful}]\label{thm:intratio}
Let $f_i , g_i: \mathbb{R}\rightarrow\mathbb{R}$ ($i = 1,2$) be four functions such that $f_2 \ge0$, $ g_2 \ge 0$, and $\int \vert f_i g_j \vert du < \infty$ ($i,j=1,2$). If $f_1/f_2$ and $g_1/g_2$ are monotonic in the opposite direction, then
\begin{equation}\label{eq:clminterg}
	\int f_1 g_1 du \int f_2 g_2 du \le \int f_1 g_2 du \int f_2 g_1 du.
\end{equation}
If $f_1 / f_2 =$ constant  or $g_1 / g_2 =  $ constant then there is equality in Eq. (\ref{eq:clminterg}).
\end{theorem}

\begin{proof}[Proof of Theorem \ref{thm:rconvex}]

Let us denote by $\tilde \rho (s,n)$ the susceptibility density of the $n$th infected individuals (i.e., the node who \textit{was infected} at step $n$).
Recall that  $\rho(s,n)$ is the density function of the susceptibility of the (susceptible) population at step $n$. 
Since the likelihood of getting infected is {proportional} to the susceptibility value of the node, $s$, it holds that:
\begin{equation}\label{eq:tildedef}
	\tilde \rho (s,n) = \frac {\rho (s,n) \cdot s} {\int  \rho ( \sigma , n) \cdot \sigma d \sigma} 
\end{equation}

According to the derivation of Lemma \ref{lem:rdef} (Eq. (\ref{eq:rdef})), the contribution to $R(n)$ of a susceptible node $a$ is approximately $S(a) \cdot I(a)$. 
Thus, the contribution of the $n$th \textit{infected} node to $R(n)$ is the same as $\int \rho(s,n) s \varphi(s) ds$ from Eq. (\ref{eq:rdef}), but where $\tilde{\rho}(s,n)$ replaces $\rho(s,n)$. Further, this contribution is exactly $R(n)-R(n+1)$. Hence, we have:  
\begin{equation}\label{eq:rdiff}
	R(n)-R(n+1) 
	=
	\int \tilde \rho (s, n) \cdot s \cdot\varphi(s) ds
\end{equation}
Plugging  Eq. (\ref{eq:tildedef}) into Eq. (\ref{eq:rdiff}) yields: 
\begin{equation}\label{eq:rdiff2}
	R(n)-R(n+1) 
	=
	\int \frac{ \rho (s,n)  s } {\int \rho ( \sigma , n)  \sigma d \sigma}  s \varphi( s ) ds
	=
	\frac
	{\int \rho ( s , n)  s^2  \varphi(s) ds}
	{\int \rho ( s , n)  s d s}
\end{equation}
Per Lemma \ref{lem:ratiodec}, it holds that:
\begin{equation}\label{eq:moninc1}
	\frac{\rho(s,n+1)}{\rho(s,n)}\text{ is monotonically non-increasing in \ensuremath{s}}.
\end{equation}
In addition, since we assumed that $\varphi(s)$ is monotonically non-decreasing in $s$ (recall Subsection 2.1),
\begin{equation}\label{eq:moninc2}
	\frac{s^2 \cdot \varphi(s)}{s} = s \cdot \varphi(s) \; \text{  is monotonically increasing in \ensuremath{s}}.
\end{equation}
Note that for any $s,n$: $s\ge0$ and $\rho(s,n)\ge0$.
Hence, according to Eq. (\ref{eq:moninc1}) and (\ref{eq:moninc2}), and by Theorem \ref{thm:intratio}, it holds that
\begin{equation}
	\int \rho(s,n+1)  s^2 \varphi(s) ds \int \rho(s,n)  s ds \le \int \rho(s,n+1)  s ds \int \rho(s,n)  s^2 \varphi(s) ds.
\end{equation}
or
\begin{equation}
	\frac{\int \rho(s,n+1)  s^2 \varphi(s) ds}{\int \rho(s,n+1)  s ds} \le \frac{\int \rho(s,n)  s^2 \varphi(s) ds}{\int \rho(s,n)  s ds}   
\end{equation}
Following Eq. \ref{eq:rdiff2}, it holds that
$$R(n+1) - R(n+2) \le R(n)-R(n+1)$$
and we conclude the proof.

\end{proof}

\section{Numerical Evaluations and Applications}\label{secvaccinetime}
In order to demonstrate the effect of the heterogeneity level on the decay process, Figure \ref{fig:R_for_Ks}  depicts the decay process for various levels of heterogeneity. The progression of $R(n)$ is illustrated for various values of $k$, the shape parameter of a spreading Gamma distribution. Gamma distributions with finite shape $k>0$ and finite scale $\theta > 0$ are frequently used to model various real-world phenomena, including the spread of COVID-19 \cite{endo2020estimating}.

\begin{figure}[h]%
    \includegraphics[width=15cm]{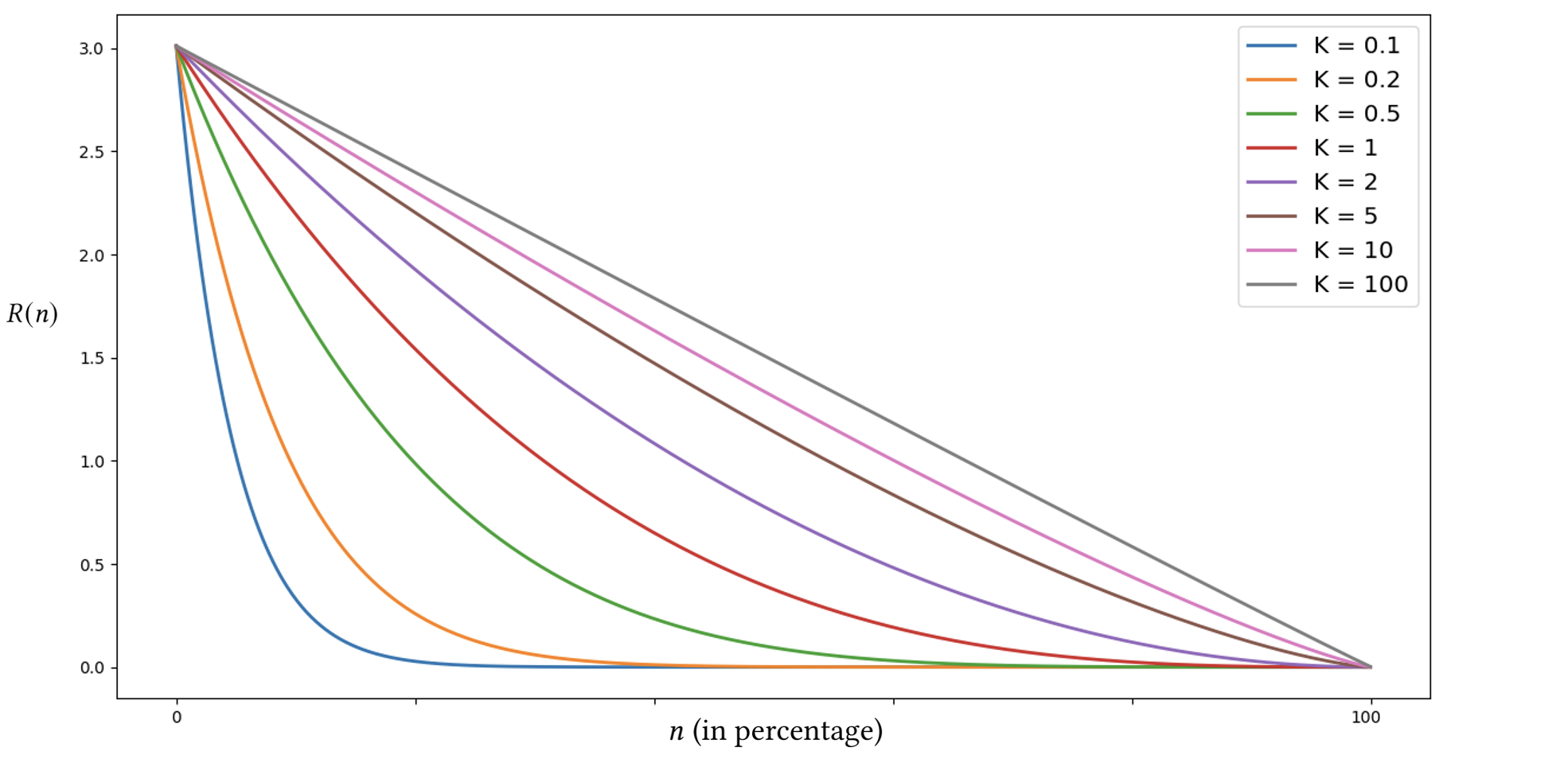}
    \caption{$R(n)$ progression for various values of $k$, the shape parameter of a spreading Gamma distribution, for $R_0 = 3$.}
	\label{fig:R_for_Ks}%
\end{figure}

{As can be seen, and in agreement with our results, the decay process remains convex for any value of $k$ (and approaches linearity when the heterogeneity level is low). Yet, the spreading heterogeneity level affects dramatically the sharpness of the reduction.
Thus, we aim at examining the sensitivity of the Herd Immunity Threshold (HIT) to countermeasure policies, in light of heterogeneity.}
\\

\noindent
{Next we use numerical examples to demonstrate  the effect of the convexity property on pandemic countermeasures. }

\paragraph{Vaccinations}
We consider the following problem: A limited supply of vaccines should be allocated to different regions (e.g., countries or different regions/states in a country, etc.). The operator's goal is to  minimize the total number of infected individuals prior to reaching herd immunity. Each region has its own spreading functions; In particular, different regions might have different heterogeneity levels.
We demonstrate that heterogeneity and the convexity of $R$ affect the outcome of such a vaccination process, as allocations which disregard the variability of heterogeneity across the regions, and conducts optimization by allocating the vaccines based only on the other parameters, namely population sizes $N_0$ and basic reproduction numbers $R_0$, might perform significantly worse in comparison to the heterogeneity-accounting optimizations.
Using numerical examples, we demonstrate its potential impact and the sensitivity of heterogeneous systems to vaccination policy.

To this end, we consider a bi-regional system in which the regions differ in their heterogeneity level (and all other parameters are equal, $N_0$ and $R_0=3$). One region consists of a homogeneous population, and one consists of heterogeneous population, whose spreading values are drawn from a Gamma distribution. 
{As the heterogeneity level differs between the two regions, we compare two alternative vaccine allocations:} Heterogeneity-accounting optimal allocation (i.e., numerically evaluate the revenue result from each vaccine using convex optimization based on our results) and Heterogeneity-oblivious optimal allocation (that conducts optimization according to $N_0$ and $R_0$, while disregarding {the heterogeneity difference}). 

In order to inspect the sensitivity of the allocation to the heterogeneity level, we repeat this examination while changing the heterogeneity level of the heterogeneous region, by controlling the \textit{shape} parameter, $k$, of the underlying Gamma distribution, {and derive the number of individuals infected prior to the HIT}.
The results are given in Figure \ref{fig:discuss_sens}. 
\begin{figure}[h]%
	\centering
	\subfloat[\centering Infected individuals prior to HIT as a function of the heterogenity level.]{{\includegraphics[width=6.5cm]{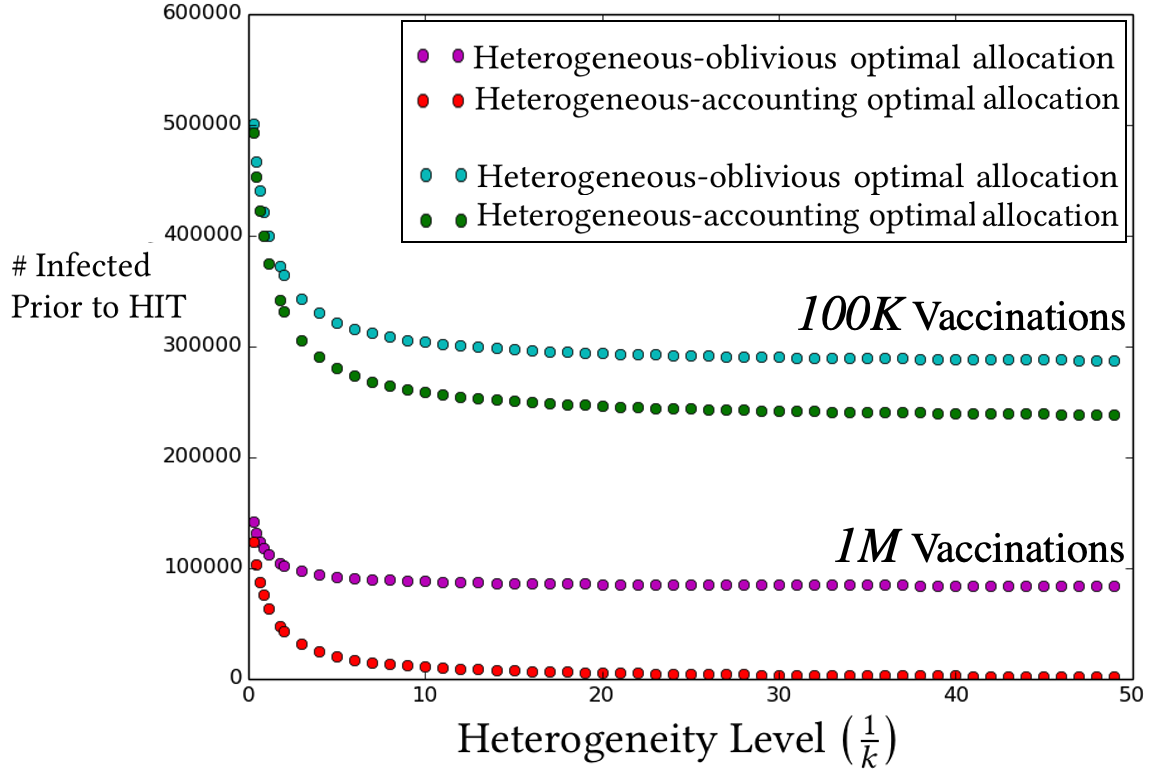} }}
	\qquad
	\subfloat[\centering Relative difference between allocation starategies as a function of the heterogenity level.]{{\includegraphics[width=6.5cm]{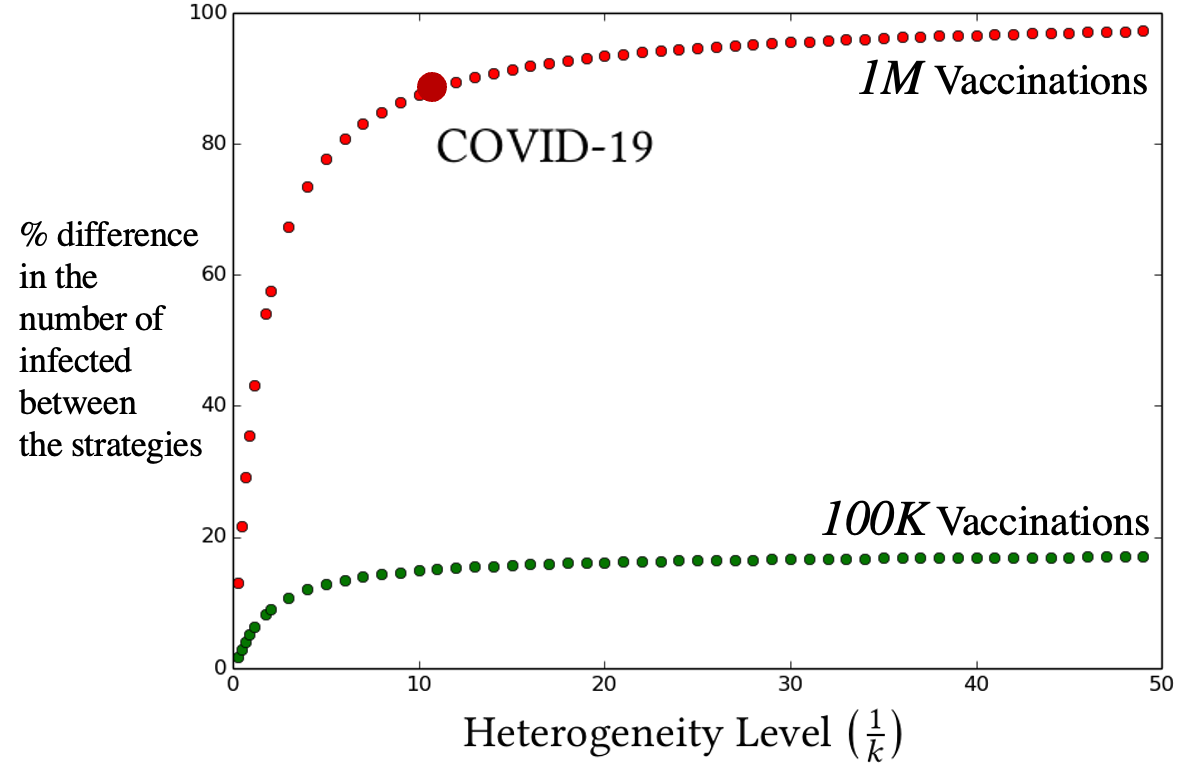} }}
	\caption{Heterogeneity effect on the system performance.}%
	\label{fig:discuss_sens}%
\end{figure}

Figure \ref{fig:discuss_sens}a depicts the number of infected individuals prior to reaching herd immunity, 
as a function of the heterogeneity level\footnote{As the shape parameter $k$ increases, the heterogeneity level decreases. Thus, we depict the results as a function of $1/k$. } under the two policies, for small number of vaccines supply ($100K$), and large number ($1M$). The red and green lines correspond to heterogeneity-accounting allocation and the purple and cyan lines correspond to the heterogeneity-oblivious allocation.
 
As can be observed, the heterogeneity-accounting allocation is significantly superior. Figure \ref{fig:discuss_sens}b depicts the {relative} difference between the number of infected (up to the HIT) under these allocations, which reaches tens of percents. The HIT is highly sensitive to the heterogeneity level of the population, and the heterogeneity-accounting optimal allocation drastically reduces the number of infected individuals prior to the HIT. In particular, for $k \approx 0.1$, which was attributed to the spreading of COVID-19 \cite{endo2020estimating}, the difference reaches a decrease of 90\% in the number of infected prior to the HIT (in comparison to the heterogeneity-oblivious optimal allocation).

Additionally, we evaluate the effect of vaccination administration timing (i.e., step in the process) on the number of individuals infected prior to reaching herd immunity. 
Such early intervention might bear counter-productive effects, as it would delay the natural process whereby super-spreaders "go out of the game" early. 
Yet, despite this observation, our simulations asserts that minimization of the HIT 
is achieved by advancing the vaccine administration timing.
The sensitivity of the HIT to the administration timing increases with the heterogeneity level of the population.

\paragraph{Lockdowns}
{In an extension} \cite{tavori2022continual} of the model used in this work, the effects of occasional spreads across networks (e.g., social gatherings and music concerts) on the spreading process were discussed.
In practice, social interactions \textbf{\textit{increase}} drastically the possibility of interacting with any random individuals within the same region or country.

It can be shown that the convexity (of $R$) result as well as the follow up applications discussed in this paper, hold for that extended model as well. This holds since the monotone likelihood ratio property of the density functions, which leads to the convexity of $R()$, is possessed under that model.
Furthermore,
\cite{tavori2022continual} 
established that the HIT is very sensitive to the lockdown type. While some lockdowns affect positively the disease blocking and decrease the HIT, others have adverse effects and might increase the HIT. 
The convexity of $R$ asserts that any restriction on social interaction  which will transform the network in to  a "semi-homogeneous" network, will increase the HIT in comparison  to the natural spread of the disease.

\section{Summary and concluding remarks}

We analyzed the stochastic process governing the viral spread in a population, and established that the decrease of the reproduction number, $R$, in heterogeneous populations is convex in the number of infected individuals.

Having established the convexity property, we addressed operational problems and examined the effect of the convexity property on infectious disease countermeasures policies, and the sensitivity of the HIT to the heterogeneity level and to the chosen policy.
We demonstrated that policies which do not account for heterogeneity might perform significantly worse (under the parameters attributed to the COVID-19 pandemic, by up to 90\%) relatively to {heterogeneity-accounting policies}. 

Further applications of the results are subject to an ongoing research.

\bibliographystyle{ACM-Reference-Format}
\bibliography{JTHL_RsConvexity_Arxiv}

\end{document}